\documentclass[a4paper,twocolumn,10pt]{quantumarticle}
\pdfoutput=1

\usepackage[utf8]{inputenc}
\usepackage[english]{babel}
\usepackage[T1]{fontenc}
\usepackage{amsmath}
\usepackage{hyperref}
\usepackage{tikz}
\usepackage{wrapfig}
\usepackage{cite}
\usepackage{amsmath,amssymb,amsfonts}
\usepackage{graphicx}
\usepackage{textcomp}
\usepackage{xcolor}
\usepackage{mathtools}
\usepackage{amsthm}
\newtheorem{theorem}{Theorem}

\usepackage[ruled,noend]{algorithm2e}

\usepackage{microtype}
\usepackage{float}
\usepackage{adjustbox}
\usepackage{booktabs,makecell,tabularx}
\usepackage{url}
\usepackage{booktabs}
\usepackage{caption}
\usepackage{subcaption}
\usepackage[braket, qm]{qcircuit}
\usepackage{braket}

\usepackage{tikz}
\usepackage{preview}
\newcounter{row}
\newcounter{col}

\definecolor{refblue}{RGB}{102, 102, 153}
\definecolor{stringgreen}{RGB}{80, 107, 65}
\definecolor{keywordblue}{RGB}{64, 89, 245}
\definecolor{commentbrown}{RGB}{59, 35, 0}

\usepackage{listings}
\lstdefinestyle{lststyle}{
  commentstyle=\color{commentbrown},
  keywordstyle=\color{keywordblue},
  numberstyle=\tiny\color{gray},
  stringstyle=\color{stringgreen},
  basicstyle=\ttfamily\footnotesize,
  breakatwhitespace=false,
  breaklines=true,
  numbers=none,
  captionpos=b,
  frame=lines,
  keepspaces=true,
  numbers=left,
  numbersep=5pt,
  showspaces=false,
  showstringspaces=false,
  showtabs=false,
  tabsize=1,
  columns=fullflexible 
}
\newcolumntype{C}[1]{>{\centering\arraybackslash}p{#1}}
\newcolumntype{L}{>{\raggedright\arraybackslash}X}
\usepackage{siunitx}
\usepackage[utf8]{inputenc}
\renewcommand{\thetable}{\arabic{table}}
\usepackage{etoolbox}
\usepackage{xparse}
\makeatletter

\newcommand{\R}{\mathbb R}
\newcommand{\Z}{\mathbb Z}

\newcommand{\C}{\mathbb C}
\renewcommand{\H}{\mathcal H}
\newcommand{\qv}{\texttt{qv}}
\newcommand{\M}{M}
\renewcommand{\P}{P}

\newcommand{\XXYY}{\mathrm{XX+YY}}

\newcommand{\pred}{\mathrm{pred}}

\SetKwFor{Control}{with control}{do}{end}
\SetKwBlock{Invert}{with invert}{end}
\SetKw{KwBy}{by}
\SetKwBlock{KwMCZ}{MCZ}{end}
\SetKw{KwCX}{CX}
\SetKw{KwX}{X}
\SetKw{KwXXYY}{XX+YY}
\SetKw{KwH}{H}
\SetKw{KwSwap}{SWAP}
\renewcommand{\tcp}[1]{\textit{/*#1}\break}
\SetKw{KwContinue}{continue} 

\begin{document}

\title{Solving the Product Breakdown Structure Problem with constrained QAOA}

\author{René Zander}
\email{rene.zander@fokus.fraunhofer.de}
\affiliation{Fraunhofer Institute for Open Communication Systems, Berlin, Germany}

\author{Raphael Seidel}
\email{raphael.seidel@fokus.fraunhofer.de}
\affiliation{Fraunhofer Institute for Open Communication Systems, Berlin, Germany}

\author{Matteo Inajetovic}
\email{m.inajetovic@tu-berlin.de}
\affiliation{Technische Universität Berlin, Berlin, Germany}

\author{Niklas Steinmann}
\email{niklas.steinmann@fokus.fraunhofer.de}
\affiliation{Fraunhofer Institute for Open Communication Systems, Berlin, Germany}

\author{Matic Petri\v{c}}
\email{matic.petric@fokus.fraunhofer.de}
\affiliation{Fraunhofer Institute for Open Communication Systems, Berlin, Germany}


\maketitle

\begin{abstract}
Constrained optimization problems, where not all possible variable assignments are feasible solutions, comprise numerous practically relevant optimization problems such as the Traveling Salesman Problem (TSP), or portfolio optimization.
Established methods such as quantum annealing or vanilla QAOA usually transform the problem statement into a QUBO (Quadratic Unconstrained Binary Optimization) form, where the constraints are enforced by auxiliary terms in the QUBO objective.
Consequently, such approaches fail to utilize the additional structure provided by the constraints.

In this paper, we present a method for solving the industry relevant Product Breakdown Structure problem. 
Our solution is based on constrained QAOA, which by construction never explores the part of the Hilbert space that represents solutions forbidden by the problem constraints. The size of the search space is thereby reduced significantly. We experimentally show that this approach has not only a very favorable scaling behavior, but also appears to suppress the negative effects of Barren Plateaus.
\end{abstract}

\section{Introduction}
\label{sec:introduction}

\begin{figure}
\includegraphics[scale = 0.55]{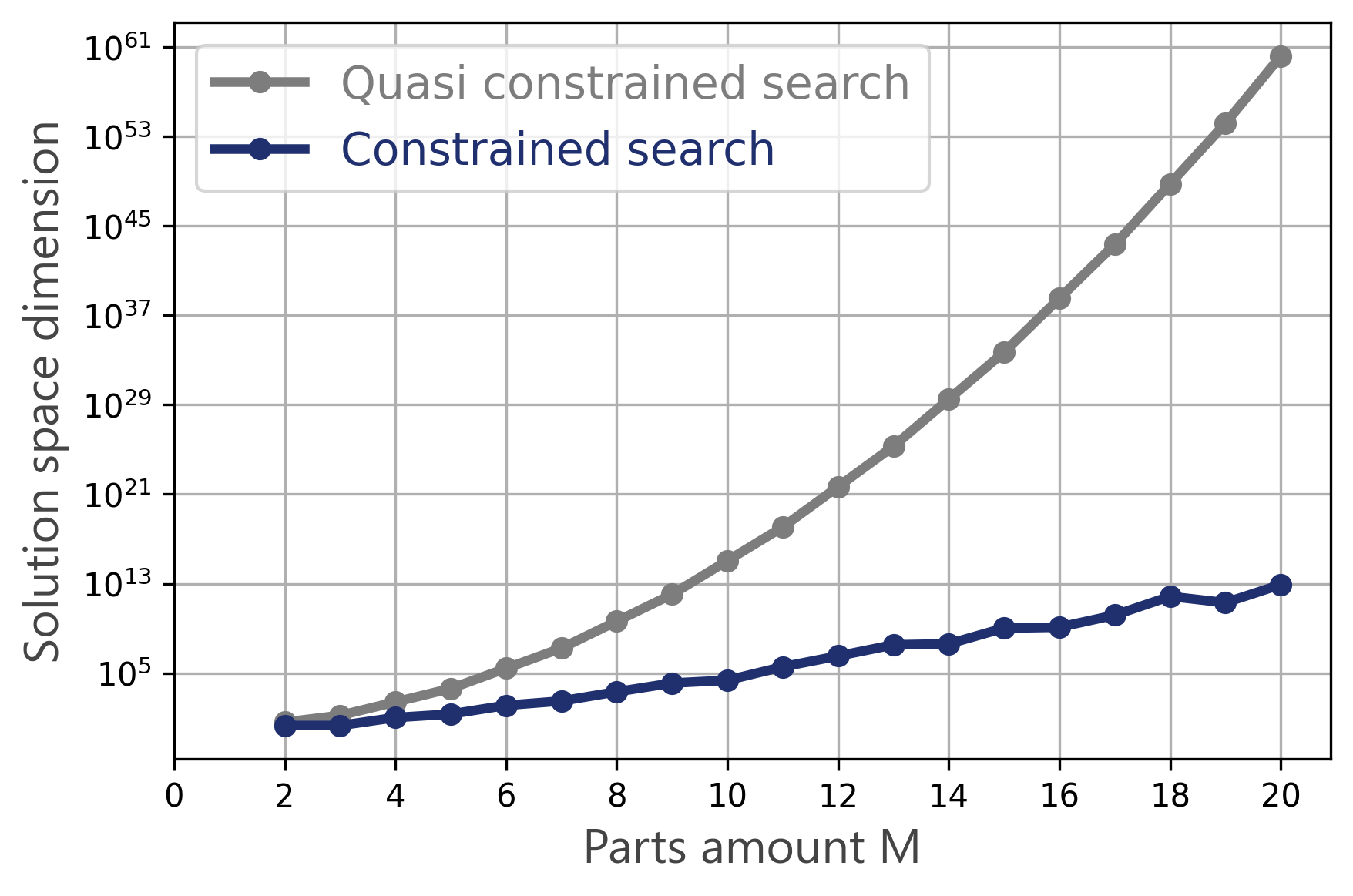}
\caption{\label{fig:state_space_dimension}Comparison of the state space dimensions of the constrained search space versus the quasi constrained search. Quasi constrained means that the full space is admitted but the problem constraints are enforced by encoding them into the cost function. The PBS graph instances, encoding the constraints, to evaluate the data are randomly generated tree graphs with maximum degree $N = \left\lfloor M/2 \right\rfloor$. 
For each $M$ we generate 30 such graphs and evaluate the average. Note that at $M = 20$ the solution space of the quasi constrained approach differs by 46 orders of magnitude, demonstrating a very clear scaling advantage.}
\end{figure}

In this paper, we develop a novel approach for solving a combinatorial optimization problem with high relevance in logistics and supply chain management, namely the product breakdown structure (PBS) problem \cite{AirbusBMW,Koopmans_1957}. A key challenge in supply chain optimization is determining how to best distribute parts manufacturing or assembly amongst possible suppliers at different geographical locations in order to minimize carbon dioxide emissions and ensure a reliable and efficient supply chain for manufacturing processes. 
The application of quantum solutions in logistics promises to accelerate the path towards sustainable and more efficient supply chains.

The PBS problem belongs to the class of constrained optimization problems. In this context, \textit{constrained} means that not every variable assignment results in a feasible solution. This class of problems comprises a wide array of practically relevant optimization problems.
While widely-used methods, such as quantum annealing or the plain vanilla Quantum Approximate Optimization Algorithm (QAOA), usually transform the problem statement into a Quadratic Unconstrained Binary Optimization (QUBO) form, we purposefully avoid this step. This is because with the conversion to a QUBO problem statement, said problem structure is usually irreversibly destroyed and can no longer be leveraged.

In this work, we utilize a modified QAOA approach, which by construction never reaches the part of the Hilbert space that represents solutions forbidden by the constraints. This, by itself, already yields powerful advantages as a result of the size of the search space being significantly reduced (see Fig.~\ref{fig:state_space_dimension}). 
This is achieved by designing a constraint-specific state preparation algorithm which creates a superposition of all feasible solutions. Thereby, a constrained mixer that leaves the subspace of feasible solutions invariant is constructed.

The outlined method is implemented in the high-level programming framework \textit{Eclipse Qrisp} \cite{Qrisp_2024,Seidel_2024,Osaba_2024}. 
\textit{Qrisp} enables developers to write quantum code with ease, facilitating a systematic approach to quantum software engineering.
Therewith, quantum computing is shifted from the domain of research and experiments to the domain of programmers, applications and businesses.

The paper is structured as follows: 
We start by defining the PBS problem in Section~\ref{sec:pbs_problem}, before providing a short overview over the Quantum Approximate Optimization Algorithm \cite{Farhi_2014} together with some of its variants addressing problem constraints \cite{Hadfield_2019,Bärtschi_2020} in Section~\ref{sec:QAOA}. 
The concept of constrained mixers is elaborated on in Section~\ref{sec:mixers}. 
In Section~\ref{sec:state_preparation}, we introduce a novel constraint-specific state preparation algorithm. Subsequently, we present compilation techniques for the constrained mixer that are specifically tailored to the used state preparation algorithm. This step is taken to ensure that the required quantum resources are reduced to a minimal amount. 
In Section~\ref{sec:implementation_benchmarking}, we cover the implementation of our method in \textit{Qrisp}, and benchmark its performance and solution quality.
A further scaling advantage is verified by a small study of our approach regarding Barren Plateau susceptibility in Appendix~\ref{appendix:barren_plateaus}.
Finally, in Section~\ref{sec:outlook}, we propose several interesting directions for further improvement of our approach in future research.

\section{Product breakdown structure problem}\label{sec:pbs_problem}

A product breakdown structure in manufacturing can be defined by 
a rooted directed tree\footnote{Here, the directed edges point from the leaves to the root.} $T$ with set of nodes $[M]=\{1,\dotsc,M\}$. 
The nodes correspond to parts with the root being the final product. 

A directed edge is given by a tuple $(r,s)$ for $r,s\in [M]$ indicating that the part $r$ is a sub-part of $s$, e.g., $r$ is a predecessor of $s$.
Let $\phi$ denote the set of (directed) edges. 
The set of predecessors of a node is denoted by $\pred(x)$. Let $d_x=|\pred(x)|$ be the degree of the node $x$, and let $m=\max_{x\in T}d_x$ the maximum degree of all nodes.

The product breakdown structure problem \cite{Koopmans_1957, AirbusBMW} can be described as follows: For a set $[N]=\{1,\dotsc,N\}$ of sites, find an assignment
\begin{equation}
f\colon[M]\rightarrow [N]
\end{equation}
that satisfies the constraints
\begin{equation}
\label{eq:constraints}
f(r)\neq f(s),\quad\text{for all } (r,s)\in\{x\}\cup\pred(x),
\end{equation}
for all nodes $x\in T$ and minimizes the cost function
\begin{equation}
\label{eq:cost}
C=\sum_{(r,s)\in\phi}c^r_{f(r),f(s)}
\end{equation}
where $c^r_{i,j}>0$ are real numbers. We assume that $c^r_{i,j}=c^r_{j,i}$ for all $r,i,j$.
Here, $c_{i,j}^r$ represents the cost (or carbon emissions) of transporting part $r$ from site $i$ to site $j$.

With this, an instance of the PBS problem is characterized by a tuple $(T,\{c^x_{i,j}\},N)$. Here, the problem constraints are encoded by the tree $T$.

The constraints (\ref{eq:constraints}) can be equivalently formulated as:
\begin{enumerate}
\item[(i)] Origin and destination of each part must be different.
\item[(ii)] The origins of any two sub-parts of a part must be different.
\end{enumerate}

Clearly, if $N<m+1$, no assignment satisfying the constraints (\ref{eq:constraints}) exists.
We denote the set of feasible solutions by $\mathcal F\subset [N]^M$. 

Define the binary variables $x_{r,i}$ for $r=1,\dotsc,M$ and $i=1,\dotsc,N$. Then the cost function can be expressed as:
\begin{equation}
\label{eq:cost_binary}
C=\sum_{(r,s)\in\phi}\sum_{i\neq j}c^r_{i,j}x_{r,i}x_{s,j}.
\end{equation}
The additional constraint
\begin{equation}
\label{eq:C1}
C_1=\sum_{r\in [M]}\left(\sum_{i}x_{r,i}-1\right)^2=0\tag{C1}
\end{equation}
ensures that each part is assigned to exactly one site. 
The constraints (i) and (ii) can be expressed as:
\begin{equation}
\label{eq:C2}
C_2=\sum_{(r,s)\in\phi}\sum_{i}x_{r,i}x_{s,i}=0.\tag{C2}
\end{equation}
and
\begin{equation}
\label{eq:C3}
C_3=\sum_{(r,s)\in\psi}\sum_{i}x_{r,i}x_{s,i}=0.\tag{C3}
\end{equation}
where $\psi=\bigcup_x\{(r,s)\mid r,s\in\pred(x)\text{ and } r<s\}$ is the set of (ordered) pairs of sub-parts of any part.
This combination of objective and constraints can be formulated as a quadratic unconstrained binary optimization (QUBO) through the addition of penalty terms to the objective with appropriate positive values for the coefficients $\lambda_1$, $\lambda_2$, and $\lambda_3$:
\begin{equation}
\label{eq:QUBO_cost}
Q=C+\lambda_1C_1+\lambda_2C_2+\lambda_3C_3.
\end{equation}

\section{Quantum Approximate Optimization Algorithm}\label{sec:QAOA}

The Quantum Approximate Optimization Algorithm (QAOA) is a hybrid quantum-classical variational algorithm designed for solving combinatorial optimization problems \cite{Farhi_2014}. The quantum and classical components work together iteratively: The quantum computer repeatedly prepares a parameter dependent quantum state and measures it, producing a classical output; this output is then fed into a classical optimization routine which produces new parameters for the quantum part. This process is repeated until the algorithm converges to an optimal or near-optimal solution.

As a first step, an optimization problem is formulated as
\begin{equation}
    z=\text{argmin}_{x\in D}C(x),
\end{equation}
where $C(x)\colon D\rightarrow \R$ is the cost function acting on $N$
variables $x=(x_1,\dotsc,x_N)$, and $D\subset\Z^N$. Typical choices are $D=\{0,1\}^N$ or $D=\{0,\dotsc,k\}^N$ (graph coloring with $k$ colors). The domain $D$ is encoded by computational basis states of a system of qubits with Hilbert space $\H$.

The QAOA operates in a sequence of layers, each consisting of a problem-specific operator and a mixing operator. To be precise, an initial state $\ket{\psi_0}$ is evolved under the action of $p$ layers of QAOA, where one layer consists of applying the unitary phase separating operator
$$U_{\P}(C,\gamma)=e^{-i\gamma C}, $$ 
which applies a phase to each computational basis state based on its cost function value; 
and the unitary mixer operator 
$$U_{\M}(B,\beta)=e^{-i\beta B}, $$ 
where $B$ represents a specific mixer Hamiltonian that drives the transitions between different states. Depending on the unitaries' eigenvalues, the QAOA parameters are typically bound to hold values $\gamma\in\{0, 2\pi\},$ and $\beta\in\{0,\pi\}$, which are then optimized classically. 

After $p$ layers of QAOA, we can define the angle dependent quantum state
\begin{equation}
\ket{\psi_p} = U_{\M}(\beta_p)U_{\P}(\gamma_p)\cdots U_{\M}(\beta_1)U_{\P}(\gamma_1)\ket{\psi_0}, 
\end{equation}
where $U_{\M}(\beta)=U_{\M}(B,\beta)$ and $U_{\P}(\gamma)=U_{\P}(C,\gamma)$.

The ultimate goal in QAOA is to optimize the variational parameters $\beta_1,\dotsc,\beta_p$ and $\gamma_1,\dotsc,\gamma_p$ in order to minimize the expectation value of the cost function $C$ with respect to the final state $\ket{\psi_p}$. This is achieved using classical optimization techniques.

For constrained optimization problems the set of feasible solutions $\mathcal F\subset\H$ is a strict subset of the Hilbert space $\H$. 
In this case, reducing the search space by exploring only solutions that satisfy the constraints of a given problem should yield better results.
This is addressed in the scope of the Quantum Alternating Operator Ansatz \cite{Hadfield_2019}.
By design, the initial state is trivial (constant depth) to implement, and the mixer is required to preserve the feasible subspace.
Hadfield et.~al.~\cite{Hadfield_2019} provide an extensive list to such mixers for a variety of optimization problems including Max-$\kappa$-Colorable Subgraph or the Traveling Salesman Problem.

In this paper, we follow a slightly different strategy \cite{Bärtschi_2020,Matsuo_2023}: 
We construct an algorithm for preparing a uniform superposition of all feasible states. That is, this algorithm implements a unitary $U_F$ such that $\ket{\psi_F}=U_F\ket{0}$ where
$$\ket{\psi_F}=\frac{1}{\sqrt{\mathcal F}}\sum\limits_{x\in\mathcal F}\ket{x}.$$
Then the initial state for the QAOA is set to $\ket{\psi_0}=\ket{\psi_F}$.
As explained in Section~\ref{sec:mixers} below, given such a state preparation algorithm, a mixer that preserves the feasible subspace can be constructed. 

Intuitively, the former approach starts in some isolated solution and transitions to the optimal solution through successive mixer applications, while the latter starts in a superposition of all feasible solutions and concentrates on the optimal solution through successive mixer applications.
An interesting study would be a comparison of the two methods for different constrained optimization problems such as the TSP.

\section{Constrained Mixers}
\label{sec:mixers}
A major hurdle of solving problems like the PBS problem using a QUBO solver are the constraints. Even though constraints can be encoded by augmenting the cost function by a term of the form $\lambda F(x)$ (where $F(x) = 1$ if $x$ is forbidden by the constraints and $0$ otherwise), this approach brings several disadvantages with it that make it very impractical when it comes to solving actual problems:

\begin{itemize}
    \item Compared to a classical algorithm, the quantum algorithm has a much larger search space because it also has to search through the forbidden states.
    \item It is not clear how to choose $\lambda$. This parameter needs to be big enough, such that forbidden states are effectively suppressed but small enough such that the hardware precision can still “resolve” the contrast in the actual cost function.
    \item Forbidden states have a non-zero probability of appearing as a solution, reducing the overall efficiency of the algorithm.
\end{itemize}

An interesting technique to overcome these problems is to encode the constraints into the mixer, such that the algorithm never leaves the ``allowed'' space. This idea has been proposed in \cite{Bärtschi_2020} and will be further refined within this work.
Let
\begin{align}
    F: \mathbb{F}_2^n \rightarrow \mathbb{F}_2,\; x \rightarrow F(x)
\end{align}

represent an arbitrary constraint function. We say $x \in \mathbb{F}_2^n $ is allowed if $F(x) = 1$ and $x$ is forbidden otherwise.

We define the uniform superposition state
\begin{align}
    \ket{\psi_F} = \frac{1}{\sqrt{N_F}} \sum_{F(x) = 1} \ket{x}
\end{align}
where $N_F=|\{x\mid F(x)=1)\}|$.

Assume that there is a quantum circuit $U_F$ preparing $\ket{\psi_F}$ (a procedure for compiling these efficiently for the PBS problem will be presented in the next section):
\begin{align}
    U_F \ket{0} = \ket{\psi_F}
\end{align}
    
We now conjugate a multi-controlled phase gate controlled on the $\ket{0}$ state with $U_F$. This yields the following unitary:
    
\begin{align}
\notag U_{\M}(\beta) &= U_F \text{MCP}_0(\beta) U_F^\dagger \\
\notag &= U_F (1 - (1- \text{exp}(i \beta)) \ket{0}\bra{0}) U_F^\dagger \\
\notag &= 1 - (1- \text{exp}(i \beta)) (U_F \ket{0})(U_F \ket{0})^\dagger \\
&= 1 - (1- \text{exp}(i \beta)) \ket{\psi_F} \bra{\psi_F}
\end{align}
    
This quantum circuit satisfies the following properties, which classify it as a valid constrained QAOA mixer

\begin{itemize}
\item $U_{\M}(\beta) \ket{x} = \ket{x}$ if $F(x) = 0$ (follows directly from $\braket{\psi_F|x}  = 0$). This property makes sure that forbidden states are mapped onto themselves, guaranteeing that the mixer only mixes among the allowed states.
\item $U_{\M}(0) = 1$. This property ensures that there is indeed no mixing happening at $\beta = 0$.
\item $|\bra{x} U_{\M}(\beta) \ket{x}| \neq 1$ for $F(x) = 1, \beta \in (0, 2\pi)$. This property shows that there is indeed \textbf{some} mixing happening for allowed states at $\beta \neq  0 $.
\end{itemize}

\subsection{Benefits of constrained mixers}

As it turns out, the constrained mixer approach improves the optimization process significantly.
Firstly, since the state space dimension is heavily reduced (compared to encoding the constraints in the cost function), the algorithm has to search through a much smaller solution space. This naturally improves the efficiency and scalability.
Furthermore, the obtained state space reduction also mitigates the well-known Barren Plateau phenomenon \cite{McClean_2018}, which depends on the Hilbert space dimension and negatively affects a lot of Variational Quantum Algorithms. 
Consider a generic parameterized quantum circuit:

 \[
    U(\boldsymbol{\theta}) = \prod_{k=1}^L U_k(\boldsymbol{\theta}_k)
 \]
 acting on a Hilbert space $\mathcal{H}=(\mathbb{C}^2)^{\otimes n}$ and a cost function   $C(\boldsymbol{\theta})=Tr[OU(\boldsymbol{\theta})\rho U^{\dagger}(\boldsymbol{\theta})]
$, where $O$ is a Hermitian observable, $\boldsymbol{\theta}$ the set of parameters, $L$ the number of layers, and $\rho$ the initial state. Barren plateaus manifest when $\mathrm{Var}_{\theta}(\partial_{\theta}C(\theta))$
vanishes exponentially with the number of qubits, and thus, the problem size.
In Appendix~\ref{appendix:barren_plateaus} we show some experiments to visualize how the proposed methodology mitigates the exponential vanishing of $\mathrm{Var}_{\theta}(\partial_{\theta}C(\theta))$ and thus, Barren Plateaus. Therefore, we demonstrate the connection between the state space reduction achieved with the constrained mixers and the trainability of the variational algorithm.

\section{State preparation}
\label{sec:state_preparation}

In this part, we present a recursive algorithm for preparing a superposition state of all feasible solutions 
to a PBS problem defined by a tree $T$, and number of sites $N$. This algorithm generalizes the methods for preparing permutation states -- for solving the TSP -- proposed in \cite{Bärtschi_2020,Matsuo_2023}.

The state is represented by a \texttt{QuantumArray q\_array} that consists of $M$ \texttt{QuantumVariables} representing the assigned site for each part as a one-hot encoded integer. That is, each such \texttt{QuantumVariable} consists of a register of $N$ qubits, and sites are encoded by the states where exactly one qubit is $\ket{1}$. This corresponds to an $N$-qubit $W$-state $\ket{W_N}=\frac{1}{\sqrt{N}}(\ket{10\dotsc 0}+\ket{010\dotsc 0}+\ket{00\dotsc 1})$. The total number of qubits required is $M\cdot N$.
In the absence of the additional constraints defined by the tree $T$, a superposition of all feasible solutions corresponds to a tensor product $\bigotimes_{r\in [M]}\ket{W_N}$. In the following, we incorporate the constraints defined by the tree $T$.

As a subroutine, we utilize the method $W(\texttt{qv},n)$ that takes a \texttt{QuantumVariable} of $N$ qubits and an integer $n\leq N$ as inputs, and prepares a partial $W$-state. That is, the first $n$ qubits of $\texttt{qv}$ are in the state $\ket{W_n}$ and the last $N-n$ qubits are in state $\ket{0}$. This method can be implemented as described in Algorithm~\ref{alg:partial_W} with $\mathcal O(n)$ gates in depth $\mathcal O(n)$. Note that Algorithm~\ref{alg:partial_W} employs the $\XXYY$-gate which acts as a ``continuous swap'' depending on the angle $\phi$.
This method also enables the preparation of non-uniform $W$-states depending on the angles $\phi_i$.
With this, non-uniform superpositions of feasible solutions can be prepared. 

Considering this, the state preparation is described in Algorithm~\ref{alg:prepare_pbs}. It employs Algorithm~\ref{alg:add_predecessors}, \ref{alg:make_distinct} as subroutines. The desired superposition state of all feasible solutions is produced recursively:
Algorithm~\ref{alg:prepare_pbs} initializes the \texttt{QuantumVariable} corresponding to the root of the tree $T$ in state $\ket{W_N}$, and subsequently applies Algorithm~\ref{alg:add_predecessors} to the root. Algorithm \ref{alg:add_predecessors}, when applied on a node $x\in T$, initializes the \texttt{QuantumVariables} corresponding to its predecessors $y\in\pred(x)$ in partial $W$-states of decreasing size, and employs Algorithm \ref{alg:make_distinct} to produce an entangled state of the initialized \texttt{QuantumVariables} that satisfies the constraints (i) and (ii). Finally, Algorithm~\ref{alg:add_predecessors} is called recursively on the nodes $y\in\pred(x)$.

A first inspection shows that this method can be implemented with $\mathcal O(M^2\cdot N)$ gates and $\mathcal O(M^2\cdot N)$ depth, with the dominant contribution resulting from enforcing the constraints (ii). 
A detailed analysis of this step reveals that (similarly to in \cite{Bärtschi_2020}) stair-shaped groups of non-overlapping controlled swap gates can be executed in parallel reducing the actual circuit depth to $\mathcal O(M\cdot N)$. Additionally, experiments indicate that the circuit depth after processing by the \textit{Qrisp} compiler might even scale linearly, i.e., as $\mathcal O(N+M)$. 

\begin{algorithm}[h]
\caption{Prepare PBS state}\label{alg:prepare_pbs}

\KwIn{A PBS tree $T$, the root $root$, the number of sites $N$, and a \texttt{QuantumArray} \texttt{q\_array} in state $\ket{0}$.}
\KwOut{A uniform superposition of all feasible solutions.}

\tcp{Prepare $W$-state for root}
$W(q\_array[root], N)$\;

\tcp{Recursively add predecessors}
Alg.~\ref{alg:add_predecessors}$(T, root, N, \texttt{q\_array})$\;

\KwRet \texttt{q\_array}
\end{algorithm}

\begin{algorithm}[h]
\caption{Add predecessors}\label{alg:add_predecessors}

\KwIn{A PBS tree $T$, a node $x$, the number of sites $N$, and a \texttt{QuantumArray} \texttt{q\_array}.}

\tcp{Sorted list of predecessors of node $x$}
$P \gets \pred(x)$\;
$m \gets |P|$\;

\tcp{Prepare partial $W$-states for predecessors}
\For{$i=0$ \KwTo $m-1$}{
    $W(q\_array[P[i]], N-1-i)$\;
}

\tcp{Create entangled state of predecessors that satisfies constraints (ii)}
\For{$i=m-2$ \KwTo $0$ \KwBy $-1$}{
    \For{$j=i+1$ \KwTo $m-1$}{
        Alg.~\ref{alg:make_distinct}$(\texttt{q\_array}[P[i]],\texttt{q\_array}[P[j]],N-2-i)$\;
    }
}

\tcp{Create entangled state of predecessors and $x$ that satisfies constraints (i)}
\For{$i=0$ \KwTo $m-1$}{
    Alg.~\ref{alg:make_distinct}$(\texttt{q\_array}[x],\texttt{q\_array}[P[i]],N-1)$\;
}

\tcp{Recursively add predecessors}
\For{$i=0$ \KwTo $m-1$}{
    Alg.~\ref{alg:add_predecessors}$(T, P[i], N, \texttt{q\_array})$\;
}
\end{algorithm}

\begin{algorithm}[h]
\caption{Make distinct}\label{alg:make_distinct}

\KwIn{\texttt{QuantumVariables qv1, qv2} in states $W_{n+1}\otimes\ket{0}$, $W_n\otimes\ket{0}$, respectively, and the index $n$.}
\KwOut{An entangled state $\sum_{i\neq j}^n\ket{\underline{i}}_1\otimes\ket{\underline{j}}_2$, where $\ket{\underline{i}}_1$ is the state of \texttt{qv1} with exactly one $1$ at position $i$.}

\For{$i=0$ \KwTo $n-1$}{
    \Control{$\texttt{qv1}[i]$}{
        \KwSwap $\texttt{qv2}[n]$, $\texttt{qv2}[i]$\;
    }
}
\KwRet \texttt{qv1, qv2}
\end{algorithm}

\begin{algorithm}[h]
\caption{Partial $W$-state}\label{alg:partial_W}

\KwIn{\texttt{QuantumVariable qv} in state $\ket{0}$, and an index $n$.}
\KwOut{A partial $W$-state $W_n\otimes\ket{0}$.}

\tcp{Flip the first qubit}
\KwX \texttt{qv}[0]

\tcp{Move the 1 by ``continuous swap'' with angles $\phi_i$ to be determined}
\For{$i=1$ \KwTo $n-1$}{
    \KwXXYY $\phi_i$, $\texttt{qv}[i]$, $\texttt{qv}[0]$\;
}
\KwRet \texttt{qv}
\end{algorithm}

\subsection{Reducing mixer requirements}\label{sec:reducing_mixer_requirements}

One disadvantage of the mixer presented in Section~\ref{sec:mixers} is the fact that a very large multi-controlled phase gate is required. If the PBS state is realized on $M \cdot N$ qubits, this MCP gate would require a depth of $\mathcal{O}(\log({M\cdot N}))$ and a gate count of $\mathcal{O}(M\cdot N)$ by using the MCX technique presented in \cite{Balauca_2022}. This can be reduced by observing an interesting property of the state preparation algorithm.

Let $\H=(\C^2)^{\otimes M\times N}$ be the Hilbert space of the $M\times N$ qubit system representing the PBS. A computational basis state $\ket{x}\in\H$ is \textit{feasible} if and only if it satisfies the constraints \ref{eq:C1}, \ref{eq:C2}, \ref{eq:C3}. Let $\H^+$ denote the Hilbert space spanned by the feasible basis states, and $\H^-$ denote the Hilbert space spanned by the infeasible basis state.

Recall that a state $\ket{x}\in\H$ is represented by a \texttt{QuantumArray} $\texttt{q\_array}=[\qv_0,\dotsc,\qv_{M-1}]$ with $M$ \texttt{QuantumVariables}. We write $\ket{\qv}=\ket{\underline{i}}$ where $\ket{\underline{i}}$ is the state of $N$ qubits with exactly one $1$ at position $i$.

\begin{theorem}
\label{thm:state_init}
Let $U_{\text{prep}}\colon\H\rightarrow\H$ be the unitary corresponding to the state preparation as described in Algorithm~\ref{alg:prepare_pbs}. Let $\ket{x}\in\H$ be a computational basis state, and
\begin{align}
    \ket{y} = U_{\text{prep}} \ket{x}.
\end{align}
Then we have either $\ket{y}\in\H^+$ or $\ket{y}\in\H^-$. That is, the state $\ket{y}$ is either a superposition of feasible states or a superposition of infeasible states. In the latter case, 
this implies $\bra{z} U_{\text{prep}} \ket{x} = 0$ for any feasible state $\ket{z}$.
\end{theorem}

\begin{proof}
To see why this is true, let the \texttt{QuantumVariable} \texttt{qv} be an entry of the \texttt{QuantumArray q\_array} representing  $\ket{x}$ that contains at least a single \texttt{1}. Let $\nu(\qv)$ be the number of \texttt{1}s in \texttt{qv}.

The partial $W$-state initialization (i.e., Algorithm~\ref{alg:partial_W}) flips the first qubit and then performs some ``continuous swaps'' to move the newly created \texttt{1} around. Applying partial $W$-state initialization to \texttt{qv} will result in a superposition of states all satisfying either of the following cases: 

\begin{itemize}
    \item zero \texttt{1}s if $\nu(\qv)=1$ and $\qv[0]=1$,
    \item as least two \texttt{1}s if $\nu(\qv)>2$, or $\nu(\qv)=2$ and $\qv[0]\neq 1$, or $\nu(\qv)=1$ and $\qv[0]\neq 1$,
    \item one \texttt{1} if $\nu(\qv)=2$ and $\qv[0]=1$.
\end{itemize}

In the first two cases, the resulting entry is therefore no longer a valid one-hot encoded integer (the condition \ref{eq:C1} is violated). The remaining steps of the PBS state preparation algorithm consist mostly of swapping the qubits around within \texttt{qv}, implying the final value of \texttt{qv} is also not a valid one-hot encoded integer.

In the third case, let $s\in T$ be the node in the PBS tree $T$ corresponding to the \texttt{QuantumVariable} $\qv$, i.e., $\qv=\texttt{q\_array}[s]$. 

If $s$ is the root of $T$, the $W$-state initialization yields 
\begin{equation}
\ket{\qv}=\alpha_0\ket{\underline{0}}+\alpha_i\ket{\underline{i}}
\end{equation} 
for some $0<i<N$. This is a non-uniform $W_N$-state and does not yield any violation of constraints.

Otherwise, let $r\in T$ be the node such that $s\in\pred(r)=\{s_0,\dotsc,s_{m-1}\}$, i.e., $s=s_k$ is a predecessor of $r$. Now, let $0<i<N-1$ be the index such that $\qv[i]=1$. 
We may assume that all \texttt{QuantumVariables} corresponding to the nodes in $r\cup\pred(r)$ are either in state $\ket{0}$ or as in the third case.

If $i<N-1-k$, the partial $W$-state initialization yields 
\begin{equation}
\ket{\qv}=\alpha_0\ket{\underline{0}}+\alpha_i\ket{\underline{i}}
\end{equation} 
This is a non-uniform $W_{N-1-k}$-state and does not yield any violation of constraints.

If $i\geq N-1-k$, the partial $W$-state initialization yields 
\begin{equation}
\ket{\qv}=\ket{\underline{i}} 
\end{equation}
We may assume that $0\leq k\leq m-1$ is the smallest index such that this case occurs.

If $i\leq N-2$, let $l=N-2-i$ (in particular, $l<k$) and $\qv'=\texttt{q\_array}[s_l]$ be the \texttt{QuantumVariable} corresponding to node $s_l$. Then the partial $W$-state initialization yields
\begin{equation}
\ket{\qv'}=\sum_{a=0}^{i}\alpha_a\ket{\underline{a}}'
\end{equation}

Algorithm~\ref{alg:make_distinct} applied to $(\qv',\qv,i)$ yields
\begin{equation}
\sum_{a=0}^{i}\alpha_a\ket{\underline{a}}'\otimes\ket{\underline{i}}\rightarrow\sum_{a=0}^{i}\alpha_a\ket{\underline{a}}'\otimes\ket{\underline{a}}
\end{equation}
which violates the constraint \ref{eq:C3} (parts $s_k$ and $s_l$ must not be at the same site). Further applications of Algorithm~\ref{alg:make_distinct} corresponding to nodes $s_m$ for $m<l$ or $r$ will move the \texttt{1}s of $\qv'$ and $\qv'$ to the same position, so that the constraint remains violated. 
If $i=N-1$, we set $\qv'=\texttt{q\_array}[r]$ and the proof is similar. In this case, the condition \ref{eq:C2} is violated.

\end{proof}

Instead of applying one large MCP gate to all qubits, we use $M$ MCP gates spanning $N$ qubits (one for each array entry). 
Thereby, the depth is reduced to $\mathcal{O}(\log(N))$. 
We have to show that this mixer still preserves the space of feasible states $\H^+$.
The aforementioned circuit gives us the unitary

\begin{align}
U_{\text{tag}}(\phi)   &= \sum_{x \in \{0,1\}^{M \times N}} \exp\left(i \frac{\omega(x) \phi}{M}\right)\ket{x}\bra{x}
\end{align}
where $\omega(x)$ is the function that counts how many 0 entries $x$ contains when interpreted as a quantum array of $M$ entries.
We now define the sets:
\begin{align}
    A &= \{ x \in \{0,1\}^{M \times N} \mid \omega(x) = 0 \}\\
    B &= \{ x \in \{0,1\}^{M \times N} \mid \omega(x) = M \}\\
    C &= \{ x \in \{0,1\}^{M \times N} \mid 0 < \omega(x) < M \}
\end{align}
Rewriting $U_{\text{tag}}$ gives:
\begin{align}
    \notag U_{\text{tag}}(\phi)   &= \sum_{x \in A} \ket{x}\bra{x}\\
        \notag &+ \sum_{x \in B} \exp(i \phi) \ket{x}\bra{x}\\
        \notag &+ \sum_{x \in C} \exp\left(i \frac{\omega(x) \phi}{M}\right) \ket{x}\bra{x}\\
        \notag &= \sum_{x \in \{0,1\}^{M \times N}} \ket{x}\bra{x}\\
        \notag &+ (\exp(i \phi) - 1) \ket{0}\bra{0}\\
        \notag &+ \sum_{x \in C} \left(\exp\left(i \frac{\omega(x) \phi}{M}\right)-1\right)\ket{x}\bra{x}\\
        &= \text{MCP}_0(\phi) + V(\phi)
\end{align}
In the last step, we identify the unitary of the MCP gate with the first two terms and denote the last term by $V(\phi)$.
If we now conjugate this unitary with the state preparation, we get
\begin{align}
    U_{\text{prep}} U_{\text{tag}}(\phi) U_{\text{prep}}^\dagger = U_{\text{mixer}}(\phi) + U_{\text{prep}} V(\phi) U_{\text{prep}}^\dagger
\end{align}


Let $\ket{\psi}$ be an arbitrary superposition of states from $\H^+$, that is, 
\begin{align}
    \ket{\psi} = \sum_{z \in \H^+} \alpha_z \ket{z}
\end{align}

Since the PBS state is of such a form the following is especially true for it:
\begin{align}
    \notag & U_{\text{prep}} V(\phi) U_{\text{prep}}^\dagger \ket{\psi}\\ 
    \notag &= \sum_{z \in D} \alpha_z U_{\text{prep}} V(\phi) U_{\text{prep}}^\dagger \ket{z}\\
    &=\sum_{z \in D} \sum_{x \in C} \alpha_z \beta_x(\phi) U_{\text{prep}}\ket{x}\bra{x} U_{\text{prep}}^\dagger \ket{z}\in\H^+
\end{align}
Here, we use that by Theorem~\ref{thm:state_init}, $U_{\text{prep}}\ket{x}$ lies either in $\H^+$ or $\H^-$, and in the latter case
\begin{align}
\bra{x} U_{\text{prep}}^\dagger \ket{z}= (\bra{z} U_{\text{prep}} \ket{x})^*= 0,
\end{align}
since $z\in\H^+$.

\section{Implementation and Benchmarking}\label{sec:implementation_benchmarking}

The described method for solving the PBS problem is implemented in \textit{Qrisp} leveraging the existing QAOA functionalities \cite{Osaba_2024} including the Trotterized Quantum Annealing (TQA) \cite{Sack_2021} parameter initialization protocol. 

To summarize, the constrained QAOA exhibits the following scaling in terms of gate count and circuit depth:

\begin{itemize}
    \item The state preparation is implemented with a few lines of code as described in Section~\ref{sec:state_preparation}. As mentioned in Section~\ref{sec:state_preparation}, it requires $\mathcal O(M^2\cdot N)$ gates and circuit depth $\mathcal O(M\cdot N)$, and utilizes $M\cdot N$ qubits for encoding a superposition of all feasible solutions. 
    
    \item The implementation of the constrained mixer, i.e., the mixer operator $U_{\M}$, is shown below. Its complexity is dominated by the contribution from the state preparation that it employs.

    \item The quantum cost operator, i.e., the phase separating operator $U_{\P}$, is generated from a \texttt{SymPy} polynomial for the classical cost function \ref{eq:cost}. This operator requires $\mathcal O(M^2\cdot N)$ gates and circuit depth $\mathcal O(M\cdot N)$.
    
\end{itemize}

In total, for a QAOA with $p$ layers our algorithm utilizes $\mathcal O(M\cdot N)$ qubits and requires $\mathcal O(p\cdot M^2\cdot N)$ gates and circuit depth $\mathcal O(p\cdot M\cdot N)$.

Essentially, the method is implemented in {\it Qrisp} as follows (here, we omit showing the implementation of straightforward classical functions):

\begin{lstlisting}[language = Python, numbers = none]
from qrisp import *

# Encode solutions as one-hot encoded integers
qtype = OHQInt(N)
q_array = QuantumArray(qtype = qtype, shape = (M))

"""
    Generate a SymPy polynomial and an orderd list of its symbols for the classical cost function for a
    dictionary of cost coefficients, number of parts M, number of sites N, and PBS tree.
"""
C, symbs = cost_symp(coeffs, M, N, PBS_tree)


"""
    Generate a classical cost function that can be evaluated on a dictionary of measurement results.
"""
cl_cost = cost_function(coeffs, PBS_tree)

# Generate quantum cost operator
def cost_op(q_array, gamma):
    app_sb_phase_polynomial(q_array, C, symbs, t=-gamma)

"""
    Generate the state preparation function for a given PBS tree and number of sites N.
    This function creates a uniform superposition state of all feasible solutions when applied to q_array.
"""
init_func = pbs_state_init(PBS_tree, N)

# Inverse state preparation function
def inv_init_func(q_array):
    with invert():
        init_func(q_array)

# Define mixer operator
def mixer_op(q_array, beta):
    with conjugate(inv_init_func)(q_array):
        for i in range(len(q_array)):
            mcp(beta, q_array[i], ctrl_state = 0)

\end{lstlisting}

With this, we have all ingredients to define the QAOA problem:

\begin{lstlisting}[language = Python, numbers = none]
from qrisp.qaoa import *

# Define QAOA problem
qaoaPBS = QAOAProblem(cost_operator=cost_op,
                        mixer=mixer_op, 
                        cl_cost_function=cl_cost, 
                        init_type='tqa')
# Initialize QAOA in uniform superposition state of all feasible solutions
qaoaPBS.set_init_function(init_func)
# Number of QAOA layers
depth = 3
# Run QAOA
res = qaoaPBS.run(q_array, depth,  max_iter = 100)

\end{lstlisting}

\smallskip

We benchmark our method on a set of PBS instances \cite{AirbusBMW}. All experiments are performed with the \textit{Qrisp} simulator on an M2 MacBook Pro. 
The results are shown in Figures~\ref{fig:benchmark_1}, \ref{fig:benchmark_2}.

\begin{figure*}
    \begin{subfigure}[]{1.0\textwidth}
        \centering
        \includegraphics[scale=0.55]{"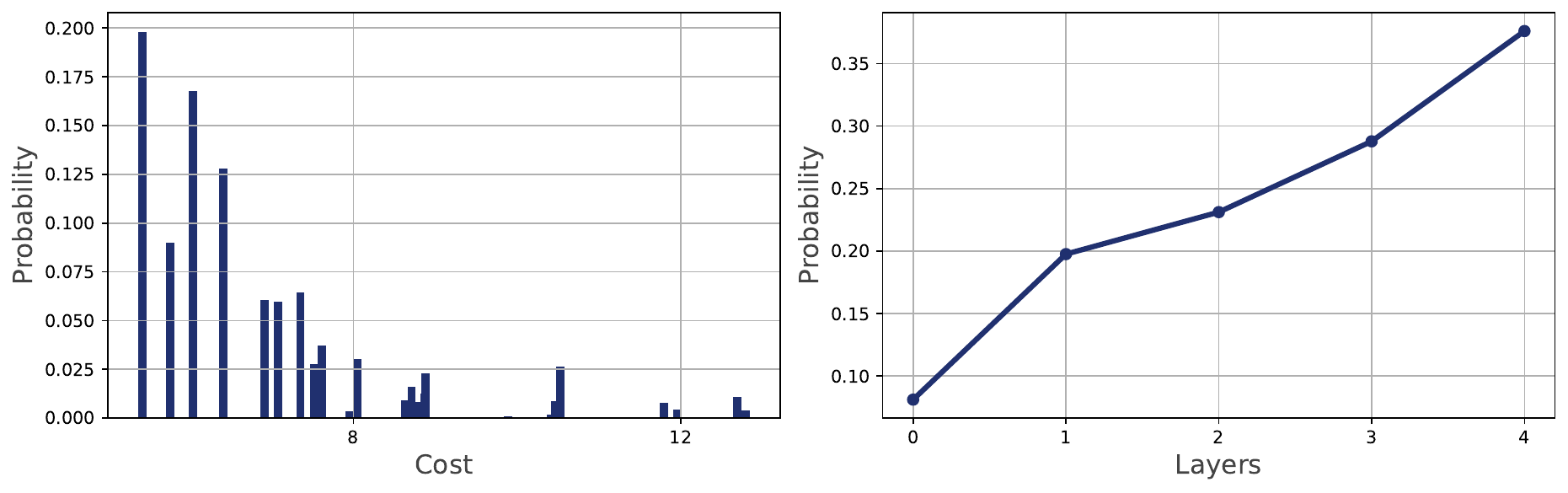"}
        \subcaption{$\phi=\{(1,0),(2,0),(3,0)\}$, $N=4$, $|\mathcal F|=24$.}
        \label{fig:benchmark_1a}
    \end{subfigure}
    
    \vspace{1em}

    \begin{subfigure}[]{1.0\textwidth}
        \centering
        \includegraphics[scale=0.55]{"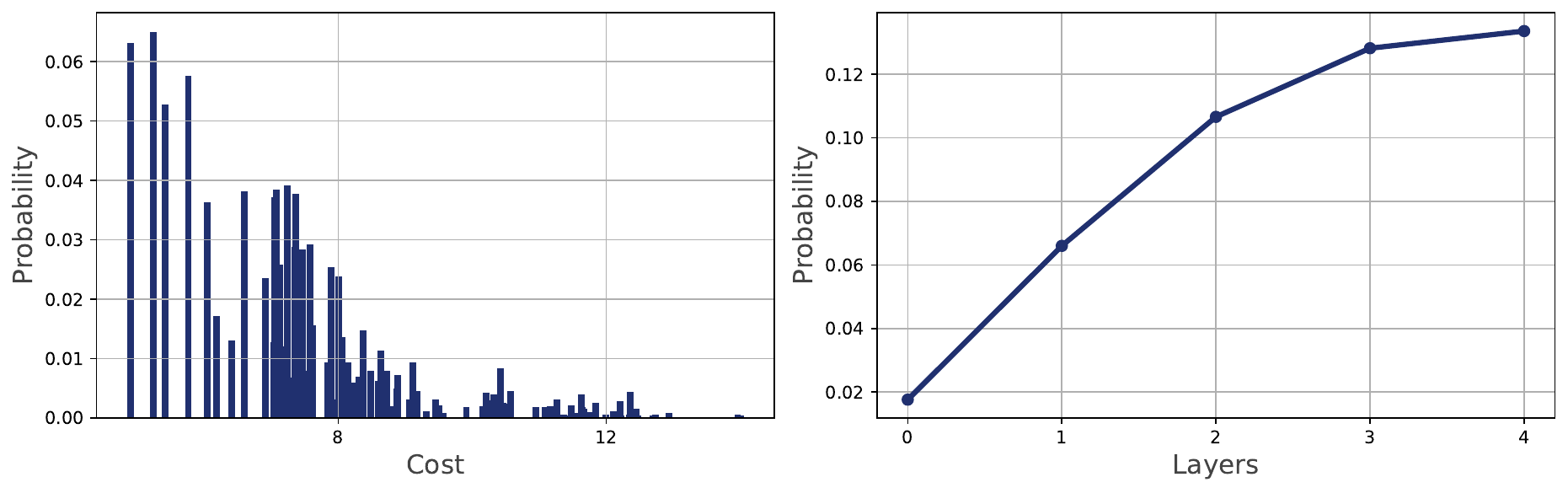"}
        \subcaption{$\phi=\{(1,0),(2,0),(3,0)\}$, $N=5$, $|\mathcal F|=120$.}
        \label{fig:benchmark_1b}
    \end{subfigure}

    \vspace{1em}

    \begin{subfigure}[]{1.0\textwidth}
        \centering
        \includegraphics[scale=0.55]{"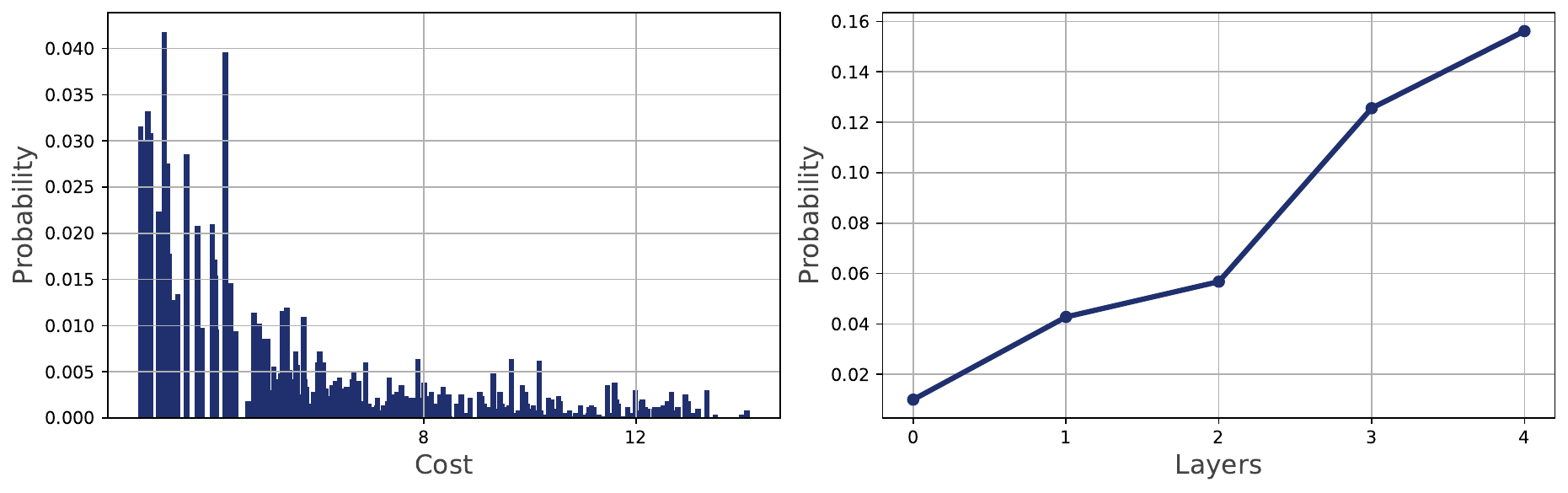"}
        \subcaption{$\phi=\{(1,0),(2,0),(3,0)\}$, $N=6$, $|\mathcal F|=360$.}
        \label{fig:benchmark_1c}
    \end{subfigure}

    \caption{Left: Aggregated probabilities of measuring a feasible solution $f$ with cost $C(f)$ for QAOA with $3$ layers. Right: Success probability $P_{0.1}$ of measuring a feasible solution $f$ with cost $C(f)<1.1\cdot C_{\min}$ for QAOA with a varying number of layers. The case $0$ layers corresponds to a uniform superposition of all feasible solutions.}
    \label{fig:benchmark_1}
\end{figure*}

\begin{figure*}
    \begin{subfigure}[]{1.0\textwidth}
        \centering
        \includegraphics[scale=0.55]{"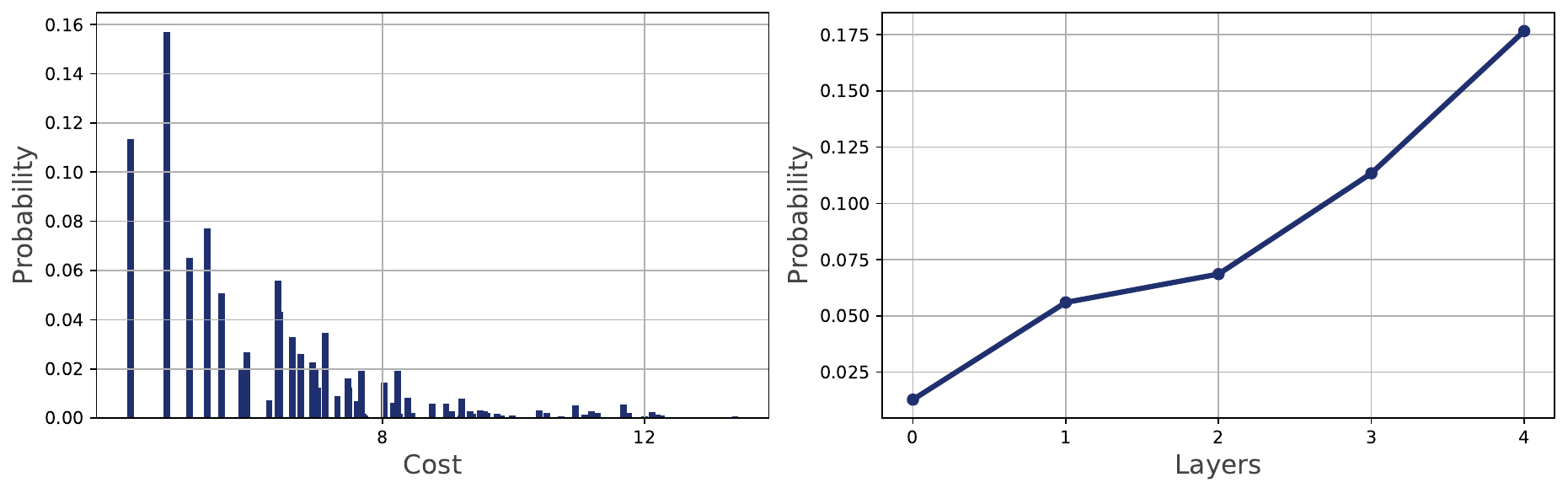"}
        \subcaption{$\phi=\{(1,0),(2,0),(3,1)\}$, $N=4$, $|\mathcal F|=72$.}
        \label{fig:benchmark_2a}
    \end{subfigure}
    
    \vspace{1em}

    \begin{subfigure}[]{1.0\textwidth}
        \centering
        \includegraphics[scale=0.55]{"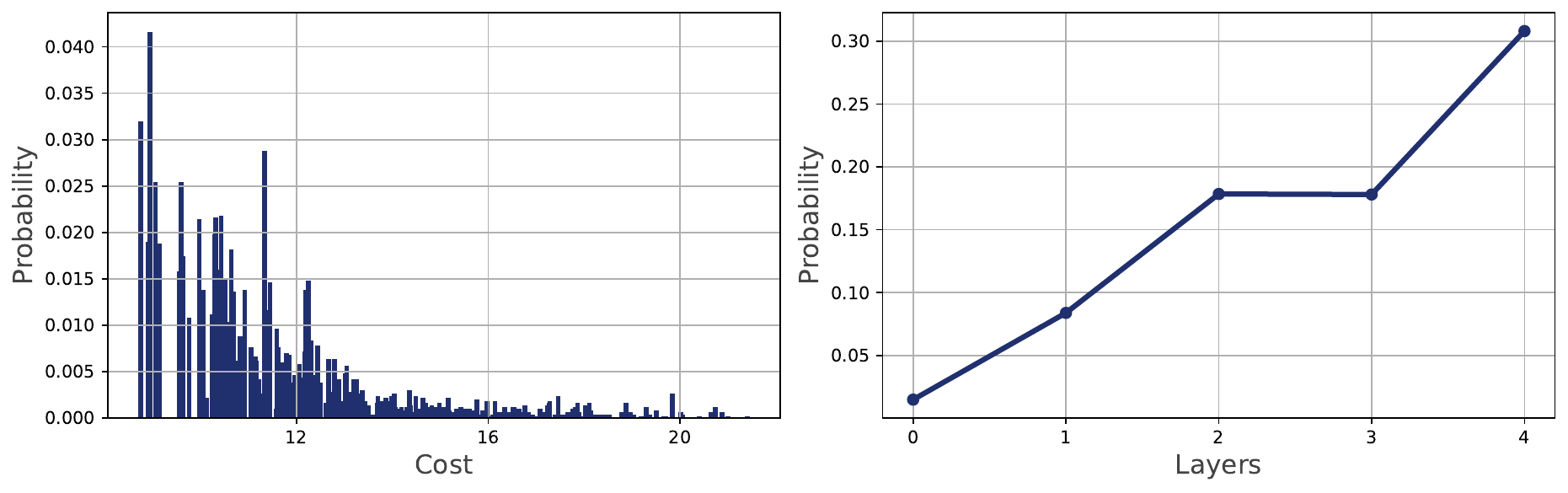"}
        \subcaption{$\phi=\{(1,0),(2,0),(3,0),(4,1)\}$, $N=5$, $|\mathcal F|=480$.}
        \label{fig:benchmark_2b}
    \end{subfigure}

    \vspace{1em}

    \begin{subfigure}[]{1.0\textwidth}
        \centering
        \includegraphics[scale=0.55]{"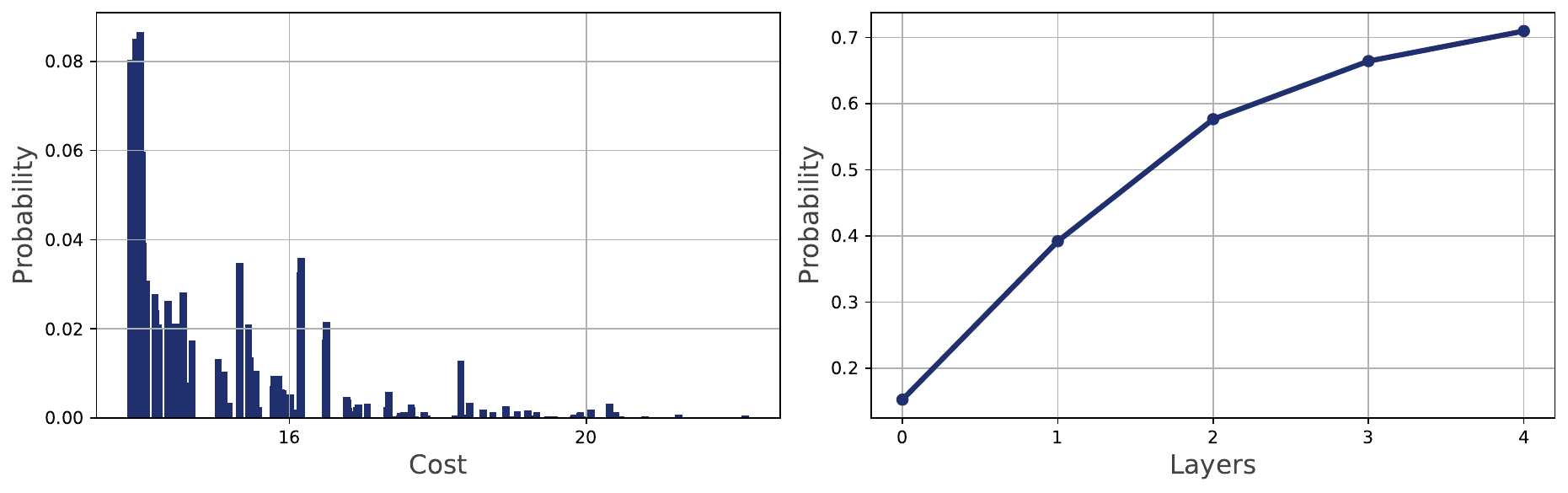"}
        \subcaption{$\phi=\{(1,0),(2,0),(3,0),(4,1),(5,1),(6,1)\}$, $N=4$, $|\mathcal F|=144$.}
        \label{fig:benchmark_2c}
    \end{subfigure}

    \caption{Left: Aggregated probabilities of measuring a feasible solution $f$ with cost $C(f)$ for QAOA with $3$ layers. Right: Success probability $P_{0.1}$ of measuring a feasible solution $f$ with cost $C(f)<1.1\cdot C_{\min}$ for QAOA with a varying number of layers. The case $0$ layers corresponds to a uniform superposition of all feasible solutions.}
    \label{fig:benchmark_2}
\end{figure*}

As a primary metric we choose - deviant form the frequently used approximation ratio (AR) - the {\bf success probability} $P_{\alpha}$ of measuring a solution $f\colon [M]\rightarrow [N]$ such that its cost value satisfies $C(f)<(1+\alpha)\cdot C_{\min}$ , for an approximation factor $\alpha>0$, where $C_{\min}$ is the cost value corresponding the optimal solution. 

The results are shown in Figures~\ref{fig:benchmark_1}, \ref{fig:benchmark_2} (right) for $\alpha=0.1$ and a varying number of layers. The findings \textbf{clearly} indicate that the QAOA routine amplifies the probabilities of (near-) optimal solutions and even scales to larger problem instances. Additionally, the aggregated probabilities of measuring a solution with a given cost are shown in Figures~\ref{fig:benchmark_1}, \ref{fig:benchmark_2}  (left). 

In general, the success probability reflects the ability of the QAOA routine to boost the probabilities of (near-) optimal solutions. We argue that this metric can be more meaningful than the overall approximation ratio of the resulting superposition state: Suppose that, on hard instances of an optimization problem of input size $n$ and solution space exponential in $n$, a QAOA routine achieves to boost the probability of measuring (near-) optimal solutions to a magnitude inverse polynomial in $n$. This would already be a potential advantage over classical algorithms, even if the probability of measuring a ``bad'' solution would be high (and therefore, the overall approximation ratio would be poor).

\section{Conclusion and Outlook}\label{sec:outlook}

In this work, we outline a strategy for solving the PBS problem with constrained QAOA on gate-based quantum computing devices. 

Our solution is implemented in the high-level programming framework \textit{Qrisp}. 
This brings the advantage of a well-arranged code base enabling rapid and seamless adoption and testing of further improvements. Potential avenues to further improve our method are the following:

\begin{itemize}

    \item Perform some numerical estimations how fast and effective the algorithm performs for practically relevant problem scales. Similar work has been conducted by Hoefler et al. in \cite{Hoefler_2023}.

    \item Develop a classification algorithm, to decide whether a PBS instance should be split using a strategy for reducing the problem size introduced below, or, alternatively, tackled with the quantum approach.

    \item Utilize the Quantum Minimum Finding algorithm \cite{vanApeldoorn_2020} to further boost the probability of measuring a (near-) optimal solution.

    \item Optimize the algorithms for specific quantum devices. For example, on ion-trap systems, the global M{\o}lmer-S{\o}rensen gate \cite{Molmer_1999} can be applied to entangle multiple qubits simultaneously. This may significantly speed up the application of the quantum cost operator $\exp(-i\gamma C)$.
    With \textit{Qrisp}, this can be conveniently implemented utilizing the \texttt{GMSEnvironment}.

    \item Develop a method to evaluate the cost function gradient without relying on the finite differences formula. In a previous work, Wierichs et al. \cite{Wierichs_2022} gave general formulas for achieving this task, which might yield a custom optimizer for this problem type.

\end{itemize}

Finally, note that the strategy described in this work is not limited to the particular problem considered here. The state preparation algorithm is specific to the constraints encoded by a tree graph, but not to the optimization problem itself which is further specified by its cost function. 
Future work may employ the blueprint outlined in this paper to other combinatorial optimization problems with constraints.

\subsection{Reducing the problem size}

In this part, we explain how the problem structure can further be leveraged so that quantum computing and classical high-performance computing can be combined for addressing industry scale problems instances.

To be precise,  we explain how a PBS instance can be decomposed, that is, how the PBS tree can be cropped into smaller instances, such that the solution of the original problem is obtained by solving smaller instances in an iterative fashion.
Depending on the structure of such instances, classical and quantum algorithms may be employed.
In Appendix \ref{appendix:classical_optimization}, we further discuss a classical dynamic programming algorithm for solving PBS problems, and thereby identify instances that are amenable to classical solvers. This is utilized to derive simple recommendations on how a PBS tree can be cropped. 

\smallskip

Decomposing a PBS instance relies on the following key observation which is a consequence of the structure of the cost function \ref{eq:cost}:
Given a tree $T$ and a node $r\in T$, let $T_r$ denote the subtree with root $r$, and $\hat T_r$ denote the tree obtained from $T$ by contracting the subtree $T_r$ to the node $r$.
An optimal solution $f$ for the PBS instance $P=(T,\{c_{i,j}^x\},N)$ restricts to an optimal solution $f'$ for the PBS instance $P'=(T_r,\{c_{i,j}^x\},N)$ with initial condition $f'(r)=f(r)$. 
Therefore, we can crop the PBS tree at the node $r$ and find an optimal solution for the problem $P$ as follows:

\begin{itemize}
    \item Solve the PBS instance $P'$ with the initial condition $f'(r)=i$ for all $i\in[N]$.
    \item Solve the PBS instance $\hat P=(\hat T_r,\{\hat c_{i,j}^x\},N)$ where $\hat c_{i,j}^x=c_{i,j}^x$ for all $x\in \hat T$ and $\hat c_{i,j}^{r}=c_{i,j}^{r}+c_i^{r}$. Here, $c_i^{r}$ is the minimal cost for the instance $P'$ given that $f'(r)=i$.
    \item Extend the optimal solution $\hat f$ for the instance $\hat P$ to an optimal solution for the instance $P$ by combining it with an optimal solution for $P'$ with initial condition $f'(r)=\hat f(r)$.
\end{itemize}

While this method comes at the expense of solving the instance $P'$ a total number of $N$ times, these steps can be performed in parallel.

An analysis of the classical algorithm presented in Appendix \ref{appendix:classical_optimization} shows that a solution can be found in runtime $\mathcal O(M\cdot N^{m+1})$, where $m$ is the maximum degree of any node in the tree $T$. In particular, for a star tree (i.e., all nodes are immediate predecessors of the root) the runtime is in $\mathcal O(N^M)$, whereas for a chain the runtime is in $\mathcal O(M\cdot N^2)$, and for a binary tree the runtime is in $\mathcal O(M\cdot N^3)$. This suggests that specific scenarios, such as star trees (or, more broadly, ‘bushy trees’) present significant challenges for classical solutions, making quantum solutions potentially more advantageous. On the other hand, there are instances that can be solved efficiently with classical methods.
This motivates the following rule for reducing the problem size: We crop the PBS tree $T$ at nodes $r$ such that the subtrees $T_r$ have a small maximum degree $m$. 
Additionally, further cropping of the tree and solving the resulting PBS instances with suitable classical or quantum algorithms can be applied.

Therewith, a suitable strategy for reducing the problem size can be developed based on industry relevant problems instances.

\section*{Acknowledgment}

This research was funded by
the Federal Ministry for Economic Affairs and Climate Action (German: Bundesministerium für Wirtschaft und Klimaschutz), projects Qompiler (grant agreement no: 01MQ22005A) and EniQmA (grant agreement no: 01MQ22007A), and by the European Union, project OASEES (HORIZON-CL4-2022, grant agreement no 101092702).
The authors are responsible for the content of this publication.

\section*{Data availability}
All data generated during this study are included in this published article.

\section*{Code availability}
\textit{Qrisp} is an open-source Python framework for high-level programming of quantum computers.
The source code is available in \href{https://github.com/eclipse-qrisp/Qrisp}{https://github.com/eclipse-qrisp/Qrisp}.
The underlying code and datasets for this study are available in 
\href{https://github.com/renezander90/PBS-QAOA}{https://github.com/renezander90/PBS-QAOA}.

\bibliographystyle{quantum}
\bibliography{sources}

\begin{thebibliography}{10}

\bibitem{AirbusBMW}
Airbus and BMW Group.
\newblock ``Quantum-powered logistics: Towards an efficient and sustainable
  supply chain''.
\newblock
  \url{https://qcc.thequantuminsider.com/wp-content/uploads/2023/12/QCChallenge_QOPT_AirbusBMWGroup_v1.1.pdf}.
\newblock Accessed: 30.04.2024.

\bibitem{Koopmans_1957}
Tjalling~C Koopmans and Martin Beckmann.
\newblock ``Assignment problems and the location of economic activities''.
\newblock Econometrica: Journal of the Econometric SocietyPages 53--76~(1957).
\newblock
  url:~\url{https://elischolar.library.yale.edu/cowles-discussion-paper-series/221/}.

\bibitem{Qrisp_2024}
Raphael Seidel, Sebastian Bock, Ren{\'e} Zander, Matic Petri{\v{c}}, Niklas
  Steinmann, Nikolay Tcholtchev, and Manfred Hauswirth.
\newblock ``Qrisp: A framework for compilable high-level programming of
  gate-based quantum computers''.
\newblock To appear~(2024).

\bibitem{Seidel_2024}
Raphael Seidel, Ren{\'e} Zander, Matic Petri{\v{c}}, Niklas Steinmann, David~Q
  Liu, Nikolay Tcholtchev, and Manfred Hauswirth.
\newblock ``Quantum backtracking in {Qrisp} applied to {Sudoku}
  problems''~(2024).
\newblock  \href{http://arxiv.org/abs/2402.10060}{arXiv:2402.10060}.

\bibitem{Osaba_2024}
Eneko Osaba, Matic Petri{\v{c}}, Izaskun Oregi, Raphael Seidel, Alejandra Ruiz,
  and Michail-Alexandros Kourtis.
\newblock ``{Eclipse Qrisp QAOA}: description and preliminary comparison with
  {Qiskit} counterparts''~(2024).
\newblock  \href{http://arxiv.org/abs/2405.20173}{arXiv:2405.20173}.

\bibitem{Farhi_2014}
Edward Farhi, Jeffrey Goldstone, and Sam Gutmann.
\newblock ``A quantum approximate optimization algorithm''~(2014).
\newblock  \href{http://arxiv.org/abs/1411.4028}{arXiv:1411.4028}.

\bibitem{Hadfield_2019}
Stuart Hadfield, Zhihui Wang, Bryan O’gorman, Eleanor~G Rieffel, Davide
  Venturelli, and Rupak Biswas.
\newblock ``From the quantum approximate optimization algorithm to a quantum
  alternating operator ansatz''.
\newblock
  \href{https://dx.doi.org/https://doi.org/10.3390/a12020034}{Algorithms {\bf
  12}, 34}~(2019).

\bibitem{Bärtschi_2020}
Andreas Bärtschi and Stephan Eidenbenz.
\newblock ``Grover mixers for {QAOA}: Shifting complexity from mixer design to
  state preparation''.
\newblock In 2020 IEEE International Conference on Quantum Computing and
  Engineering (QCE).
\newblock \href{https://dx.doi.org/10.1109/QCE49297.2020.00020}{Pages 72--82}.
\newblock ~(2020).

\bibitem{Matsuo_2023}
Atsushi Matsuo, Yudai Suzuki, Ikko Hamamura, and Shigeru Yamashita.
\newblock ``Enhancing {VQE} convergence for optimization problems with
  problem-specific parameterized quantum circuits''.
\newblock
  \href{https://dx.doi.org/https://doi.org/10.1587/transinf.2023EDP7071}{IEICE
  TRANSACTIONS on Information and Systems {\bf 106}, 1772--1782}~(2023).

\bibitem{McClean_2018}
Jarrod~R. McClean, Sergio Boixo, Vadim~N. Smelyanskiy, Ryan Babbush, and
  Hartmut Neven.
\newblock ``Barren plateaus in quantum neural network training landscapes''.
\newblock \href{https://dx.doi.org/10.1038/s41467-018-07090-4}{Nature
  Communications{\bf 9}}~(2018).

\bibitem{Balauca_2022}
Stefan Balauca and Andreea Arusoaie.
\newblock ``Efficient constructions for simulating multi controlled quantum
  gates''.
\newblock In Computational Science -- ICCS 2022.
\newblock
  \href{https://dx.doi.org/https://doi.org/10.1007/978-3-031-08760-8_16}{Pages
  179--194}.
\newblock Springer International Publishing~(2022).

\bibitem{Sack_2021}
Stefan~H Sack and Maksym Serbyn.
\newblock ``Quantum annealing initialization of the quantum approximate
  optimization algorithm''.
\newblock
  \href{https://dx.doi.org/https://doi.org/10.22331/q-2021-07-01-491}{Quantum
  {\bf 5}, 491}~(2021).

\bibitem{Hoefler_2023}
Torsten Hoefler, Thomas H{\"a}ner, and Matthias Troyer.
\newblock ``Disentangling hype from practicality: On realistically achieving
  quantum advantage''.
\newblock
  \href{https://dx.doi.org/https://doi.org/10.1145/3571725}{Communications of
  the ACM {\bf 66}, 82--87}~(2023).

\bibitem{vanApeldoorn_2020}
Joran van Apeldoorn, Andr{\'{a}}s Gily{\'{e}}n, Sander Gribling, and Ronald
  de~Wolf.
\newblock ``Quantum {SDP}-{S}olvers: {B}etter upper and lower bounds''.
\newblock \href{https://dx.doi.org/10.22331/q-2020-02-14-230}{{Quantum} {\bf
  4}, 230}~(2020).

\bibitem{Molmer_1999}
Klaus M\o{}lmer and Anders S\o{}rensen.
\newblock ``Multiparticle entanglement of hot trapped ions''.
\newblock \href{https://dx.doi.org/10.1103/PhysRevLett.82.1835}{Physical Review
  Letters {\bf 82}, 1835--1838}~(1999).

\bibitem{Wierichs_2022}
David Wierichs, Josh Izaac, Cody Wang, and Cedric Yen-Yu Lin.
\newblock ``General parameter-shift rules for quantum gradients''.
\newblock \href{https://dx.doi.org/10.22331/q-2022-03-30-677}{Quantum {\bf 6},
  677}~(2022).

\bibitem{Larocca_2022}
Martin Larocca, Piotr Czarnik, Kunal Sharma, Gopikrishnan Muraleedharan,
  Patrick~J. Coles, and M.~Cerezo.
\newblock ``Diagnosing barren plateaus with tools from quantum optimal
  control''.
\newblock \href{https://dx.doi.org/10.22331/q-2022-09-29-824}{Quantum {\bf 6},
  824}~(2022).

\end{thebibliography}

\appendix
\section{Classical optimization}
\label{appendix:classical_optimization}

A classical dynamic programming method for solving the PBS problem is described in Algorithm~\ref{alg:classical_pbs}. This method is amenable to a complexity theoretic analysis of its runtime and thereby facilitates distinguishing classically solvable instances from intractable ones.

Recall that $m=\max_{x\in T}|\pred(x)|$ is the maximal number of predecessors of any node in the tree $T$. Then Algorithm~\ref{alg:classical_pbs} finds the optimal solution in runtime $\mathcal{O}(M\cdot N^{m+1})$: For each node $x\in [M]$, and  each assignment $f(x)=i\in [N]$, it explores all feasible assignments $f(y)=j\in [N]$ for $y\in\pred(x)$. 

\begin{algorithm}[htbp]
\caption{Classical PBS solution}\label{alg:classical_pbs}

\KwIn{A PBS instance $P=(T,\{c_{i,j}^x\},N)$.}

\KwOut{The optimal solution $f\colon [M]\rightarrow [N]$.}

\tcp{Dictionary to store partial solutions}
$J \gets \emptyset$

\tcp{Let the height $h(x)$ of a node $x$ in a tree be the number of edges on the longest path from the node down to a leaf. A leaf has height of 0. The height of a tree $h(T)$ is the height of its root. Let $T_h$ be the set of nodes of height $h$.}

\For{$x\in T_0$}{   
    \For{$i=1$ \KwTo $N$}{
        $c_x^i \gets 0$\;
    }
}

\For{$h=1$ \KwTo $h(T)$}{
    \For{$x\in T_h$}{   
        \For{$i=1$ \KwTo $N$}{
            $c_i^x \gets \min\limits_{j(y)}\left\{\sum\limits_{y\in\pred(x)} (c_j^y+c_{i,j}^y)\right\}$\;
            \tcp{$j^*(y)$ the optimal feasible assignment of sites $j\in [N]$ to parts $y\in\pred(x)$ given that $x$ is at site $i$.}
            $J[x,i] \gets j^*(y)$\;
        }
    }
}

\tcp{Find optimal cost value and assignment for the root.}
$C \gets \min\limits_{i\in [N]}\{c_r^i\}$\;
Recover optimal solution $f$ from $J$

\end{algorithm}

\section{Variance of Cost Function's gradients }
\label{appendix:barren_plateaus}
As  aforementioned in Section~\ref{sec:mixers}, the proposed constrained mixer solution lightens the burden of vanishing gradients. In Figure~\ref{fig:variance_bp} we show that (increasing the problem size according to the experimental setup in Table~\ref{tab:variance_table_setup}) the decaying of the variance is less impactful compared to other variational algorithms. Specifically, in the case of the most known variational algorithms affected by Barren Plateaus, the variance becomes significantly smaller by several orders of magnitude, even with circuits with a much smaller number of qubits compared to our experiments \cite{Larocca_2022}.

\begin{table}[htbp]
    \centering
    \begin{tabular}{c|c|c}
        $N_q$  &  $N$    &  $\phi$\\\hline
        12 & 3 & $\{(1,0),(2,0),(3,1)\}$\\
        16 & 4 & $\{(1,0),(2,0),(3,1))\}$\\
        20 & 4 & $\{(1,0),(2,0),(3,1),(4,2)\}$\\
        24 & 4 & $\{(1,0),(2,0),(3,1),(4,2),(5,4)\}$\\
    \end{tabular}
    \caption{Setup for the experiments shown in Figure~\ref{fig:variance_bp}. $N_q$ is the number of qubits and $N$ the number of sites.}
    \label{tab:variance_table_setup}
\end{table}

\begin{figure*}[]
    \centering
    \includegraphics[scale=0.52]{"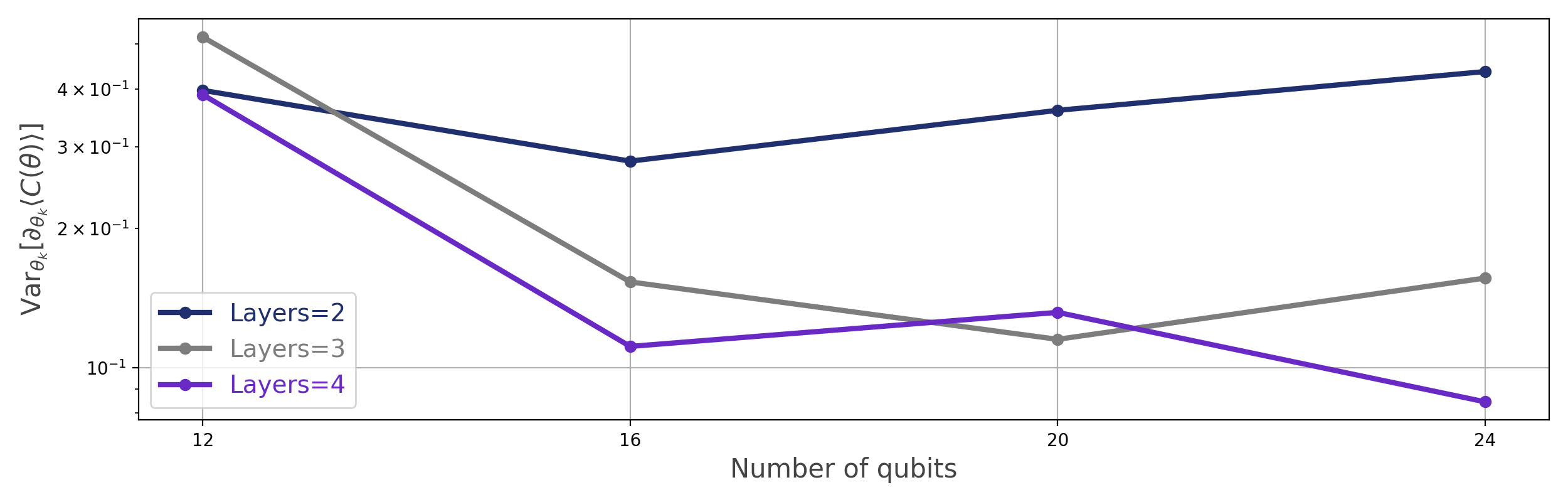"}
    \caption{Logarithmical plot visualizing the variance of the cost function's gradients of the proposed solution's Ansatz. The experiments are done increasing the problem sizes as reported in Table~\ref{tab:variance_table_setup}. The gradient of the cost function is computed with the finite difference method w.r.t. the last parameter of the Ansatz, and the variance is calculated with 100 samples.
    }
    \label{fig:variance_bp}
\end{figure*}

\end{document}